\documentclass[copyright,creativecommons]{eptcs}
\providecommand{\event}{B. Klin, P. Soboci\'nski (Eds.): \\
6th Workshop on Structural Operational Semantics (SOS'09)}
\providecommand{\volume}{19}
\providecommand{\anno}{2010}
\providecommand{\firstpage}{77}
\providecommand{\eid}{6}
\usepackage{breakurl}
\usepackage{color}
\usepackage{fancyvrb,latexsym,alltt}

\usepackage{amsmath,amssymb,amsthm}
\usepackage{color,mathpartir,cite}

\title{
  Rewriting Logic Semantics of a Plan Execution Language
\thanks{Authors in alphabetical order.}}
 
\author{
  Gilles Dowek
  \institute{
    \'Ecole polytechnique and INRIA\\
    LIX, \'Ecole polytechnique \\
    91128 Palaiseau Cedex, France
  }
  \email{gilles.dowek@polytechnique.fr}
  \and
  C\'{e}sar Mu\~{n}oz
  \institute{
    NASA\\
    Langley Research Center\\
    MS. 130, Hampton, VA 23681, USA
 } 
 \email{cesar.a.munoz@nasa.gov}
  \and
  Camilo Rocha
  \institute{
    Department of Computer Science \\
    University of Illinois at Urbana-Champaign \\
    201 Goodwin Ave, Urbana, IL 61801, USA
  }
  \email{hrochan2@illinois.edu}
}
\def\titlerunning{Rewriting Logic Semantics of a Plan Execution Language}
\def\authorrunning{G. Dowek \& C. Mu\~{n}oz \& C. Rocha}

\newcommand{\eqaci}{=_{ACUI}}
\newcommand{\lra}{\rightarrow}
\newcommand{\lrn}[1]{{\lra}^{#1}}
\newcommand{\async}{\stackrel{\Box}{\lra}}
\newcommand{\para}{\stackrel{||}{\lra}}
\newcommand{\sync}[2]{\stackrel{#2}{#1}}
\newcommand{\syncr}[1]{\sync{\lra}{#1}}

\newcommand{\rew}{\longrightarrow}

\newbox\tempa
\newbox\tempb
\newdimen\tempc
\def\mud#1{\hfil $\displaystyle{\mathstrut #1}$\hfil}
\def\rig#1{\hfil $\displaystyle{#1}$}
\def\irulehelp#1#2#3{\setbox\tempa=\hbox{$\displaystyle{\mathstrut #2}$}%
                        \setbox\tempb=\vbox{\halign{##\cr
        \mud{#1}\cr
        \noalign{\vskip\the\lineskip}
        \noalign{\hrule height 0pt}%
        \rig{\vbox to 0pt{\vss\hbox to 0pt{${\; #3}$\hss}\vss}}\cr
        \noalign{\hrule}%
        \noalign{\vskip\the\lineskip}%
        \mud{\copy\tempa}\cr}}%
                      \tempc=\wd\tempb
                      \advance\tempc by \wd\tempa
                      \divide\tempc by 2 }
\def\irule#1#2#3{{\irulehelp{#1}{#2}{#3}%
                     \hbox to \wd\tempa{\hss \box\tempb \hss}}}

\newtheorem{theorem}{Theorem}[section]

\newtheorem{proposition}[theorem]{Proposition}

\newenvironment{definition}[1][Definition]{\begin{trivlist}
\item[\hskip \labelsep {\bfseries #1}]}{\end{trivlist}}

\def\thname#1{\mathcal{#1}}
\def\plexil{\mathrm{PXL}}

\DefineVerbatimEnvironment{code}{Verbatim}{fontfamily=courier,fontsize=\scriptsize}

\begin{document}

  \maketitle

\begin{abstract}
  The {\em Plan Execution Interchange Language} (PLEXIL)
  is a synchronous language 
  developed by NASA to support autonomous spacecraft operations.
  In this paper, we propose a rewriting logic semantics of PLEXIL in
  Maude, a high-performance logical engine. The rewriting logic
  semantics is by itself a formal interpreter of the language and can
  be used as a semantic benchmark for the implementation
  of PLEXIL executives. The implementation in Maude has the additional 
  benefit of making available to PLEXIL designers and developers 
  all the formal analysis and verification tools
  provided by Maude.
  The formalization of the PLEXIL semantics in
  rewriting logic poses an interesting challenge due to the synchronous
  nature of the language and the prioritized rules defining 
  its semantics. To overcome this difficulty, we propose a 
  general procedure for simulating synchronous set relations
  in rewriting logic that is sound and, for deterministic relations, 
  complete. We also report on two issues at the design level of the 
  original PLEXIL semantics that were identified with the help 
  of the executable specification in Maude.
\end{abstract}

\section{Introduction}
\label{sec.intro}

Synchronous languages were introduced in the 1980s
to program {\em reactive systems}, i.e., systems whose behavior is
determined by their continuous reaction to the environment where they are deployed. 
Synchronous  languages are often used to program embedded applications
and  automatic control software. 
The family of synchronous languages is characterized 
by the {\em synchronous hypothesis}, 
which states that a reactive system is
arbitrarily fast and able to react immediately in
no time to stimuli from the external environment.
One of the main consequences of the synchronous
hypothesis is that components running in parallel
are perfectly synchronized and cannot arbitrarily
interleave. The implementation of a synchronous language usually
requires the simulation of  the synchronous semantics into
an asynchronous computation model. This simulation
must ensure the validity of the synchronous hypothesis in the target 
asynchronous model.

The {\em Plan Execution Interchange Language} ~(PLEXIL)~\cite{plexil}
is a synchronous language developed 
by NASA to support autonomous spacecraft operations.
Space mission operations require flexible, efficient and 
reliable plan execution. The computer system on board the spacecraft
that executes plans is called the {\em executive} and it is a safety-critical
component of the space mission. The {\em Universal Executive} ~(UE)~\cite{ue}
is an open source PLEXIL executive developed 
by NASA\footnote{\url{http://plexil.sourceforge.net}.}. PLEXIL and the UE
have  been used on mid-size applications such as 
robotic rovers and a prototype of a Mars drill, and to demonstrate automation
for the International Space Station.

Given the critical nature of spacecraft operations, PLEXIL's
operational semantics has been formally defined~\cite{DMP08} and
several properties of the language, such as determinism and compositionality,
have been mechanically
verified~\cite{DMP07ICAPSW} in the Prototype Verification System
(PVS)~\cite{ORS92}.  The formal small-step  semantics is  defined using 
a compositional  layer of five reduction relations on sets of {\em nodes}. These nodes
are the building blocks of a PLEXIL plan and 
represent the hierarchical decomposition of tasks. 
The {\em atomic relation} describes
the execution of an individual node in terms of state
transitions triggered by changes in the environment. The
{\em micro relation} describes the {\em synchronous} reduction
of the atomic relation with respect to the {\em maximal redexes
strategy}, i.e., the synchronous application of the atomic
relation to the maximal set of nodes of a plan. The remaining
three relations are the {\em quiescence relation}, the {\em macro
relation} and the {\em execution relation} which describe
the reduction of the micro relation until normalization,
the interaction of a plan with the external environment, and
the $n$-iteration of the macro relation corresponding to
$n$ time-steps, respectively. From an operational point of view,
PLEXIL is more complex than general purpose synchronous 
languages such as Esterel~\cite{Esterel00} or Lustre~\cite{LustrePOPL87}.
PLEXIL is designed specifically for flexible and reliable  
command execution in autonomy applications.


In this paper, we propose a rewriting logic semantics of PLEXIL in 
Maude~\cite{maude-book} that complements the small-step structural
operational semantics written in PVS. In contrast to the PVS higher-order logic specification, the 
rewriting logic semantics of PLEXIL is executable and it
is by itself an interpreter of the language. This interpreter is intended to be
a semantic benchmark for validating the implementation of PLEXIL executives such as the Universal
Executive and a testbed for designers of the language to study new features or possible 
variants of the language. Additionally, by using a graphical interface~\cite{RMC09}, 
PLEXIL developers will be able to exploit the formal analysis tools provided by Maude to verify properties
of actual plans.

Rewriting logic is a logic of concurrent change in which
a wide range of models of computation and logics can be
faithfully represented. 
The rewriting semantics of a synchronous language such as PLEXIL poses 
interesting practical challenges because Maude implements the maximal 
concurrency of rewrite rules by interleaving concurrency. That is, 
although rewriting logic allows for concurrent synchronous specifications at 
the mathematical level, Maude executes the rewrite rules by interleaving 
concurrency. To overcome this situation, we develop a {\em serialization procedure}
that allows for the simulation of a synchronous relation
via set rewriting systems. This procedure is presented in a library of abstract set
relations that we have written in PVS.  The procedure is sound and
complete for the {\em synchronous closure} of any deterministic relation
under the {\em maximal redexes strategy}.


We are collaborating with the PLEXIL development team at NASA Ames
by using the rewriting logic semantics of PLEXIL
to validate the intended semantics of the language against a wide 
variety of plan examples. We report on two 
issues of PLEXIL's original semantics that were discovered with the help of 
the rewriting logic semantics of PLEXIL presented in this paper:
the first was found at the level of the atomic relation for which
undesired interleaving semantics were introduced in some computations,
and the second was found at the level of the micro relation for
which spurious infinite loops were present in some computations.
Solutions to both issues were provided by the authors, and have been adopted
in the latest version of the PLEXIL semantics.

Summarizing, the contributions presented in this paper are:

\begin{itemize}
	\item The rewriting logic specification of the PLEXIL semantics.
	\item A library of abstract set relations suitable for the definition and verification
                  of synchronous relations. 
	\item A serialization procedure for the simulation of synchronous relations by
                  rewriting, and an equational version of it in rewriting logic
                  for deterministic synchronous relations.
	\item The findings on two issues in the design of 
	the original PLEXIL semantics, and the corresponding solutions that 
	were adopted in an updated version of the language semantics.

\end{itemize}

\paragraph{Outline of the paper.} Background on rewriting logic, and
the connection between this logic and Structural Operational Semantics 
are summarized in Section~\ref{sec.rl-sos}.
In Section~\ref{sec.library} we present the library of set relations, including the soundness and completeness
proof of the serialization procedure. 
Section~\ref{sec.plexilrls} describes the rewriting logic semantics of 
PLEXIL. In Section~\ref{sec.results} we discuss preliminary results.
Related work and concluding remarks are presented in Section~\ref{sec.concl}.

\section{Rewriting Logic and Structural Operational Semantics}
\label{sec.rl-sos}

Rewriting logic~\cite{unified-tcs} is a general semantic framework
that unifies in a natural way a wide range of models of
concurrency. 
Language specifications can be executed in Maude,
a high-performance rewriting logic implementation, and
benefit from a wide set of formal analysis tools
available to it, such as Maude's LTL Model Checker.

A {\em rewriting logic specification} or {\em theory} is
a tuple $\thname{R} = (\Sigma, E \cup A , R)$ where:

\begin{itemize}
	\item $(\Sigma, E \cup A)$ is a membership equational logic
	theory with $\Sigma$ a signature
	having a set of kinds, a family of sets of operators, and
	a family of disjoint sets of sorts; $E$ a set of 
	$\Sigma$-sentences, which are universally quantified
	Horn clauses with atoms that are equations $(t = t')$ and
	memberships $(t:s)$, with $t,t'$ terms and $s$ a sort; 
  $A$ a set of ``structural'' axioms (typically associativity
  and/or commutativity and/or identity) such that there
  exists a {\em matching algorithm modulo} $A$ producing
  a finite number of $A$-matching substitutions; and

	\item $R$ a set of universally quantified {\em conditional
	rewrite rules} of the form
	\[(\forall X)\; r : t \longrightarrow t'\;\mathrm{if}\; 
	  \bigwedge_i u_i = u_i' \land \bigwedge_j v_j:s_j \land
	  \bigwedge_l w_l \longrightarrow w_l'\]
	  where $X$ is a set of sorted variables, 
	  $r$ is a label, $t,t',u_i,u_i',v_j,w_l$ and $w_l'$
	  are terms with variables among those in $X$, 
	  and $s_j$ are sorts.
\end{itemize}

Intuitively, $\thname{R}$ specifies a {\em concurrent system},
whose states are elements of the initial algebra $T_{\Sigma/E\cup A}$
specified by the theory $(\Sigma, E \cup A)$ and whose
{\em concurrent transitions} are specified by the rules $R$.
Concurrent transitions are deduced according to the set of
inference rules of rewriting logic, which are described in
detail in~\cite{bruni06} (together with a precise account of the
more general forms of rewrite theories and their models). 
Using these inference rules,
a rewrite theory $\thname{R}$ proves a statement of the form
$(\forall X)\;t \longrightarrow t'$, written as $\thname{R} \vdash
(\forall X)\;t \longrightarrow t'$, meaning that, in $\thname{R}$,
the state term $t$ can transition to the state term $t'$ in
a finite number of steps. A detailed discussion of rewriting
logic as a unified model of concurrency and its inference system
can be found in~\cite{unified-tcs}.

We have a one-step rewrite $[t]_{E\cup A}
\longrightarrow_\thname{R} [t']_{E \cup A}$ in $\thname{R}$
iff we can find a term $u \in [t]_{E\cup A}$ such that
$u$ can be rewritten to $v$ using some rule 
$r : a \longrightarrow b \;{\rm if}\;C \in R$ in the
standard way (see~\cite{dershowitz90}),
denoted $u \longrightarrow_R v$, and we furthermore have
$v \in [t']_{E\cup A}$. For arbitrary 
$E$ and $R$, whether $[t]_{E\cup A} \longrightarrow_\thname{R} 
[t']_{E \cup A}$ holds is in general {\em undecidable}, 
even when the equations in $E$ are confluent and 
terminating modulo $A$. Therefore, the most useful 
rewrite theories satisfy additional 
executability conditions under which we can reduce 
the relation $[t]_{E\cup A} \longrightarrow_\thname{R} [t']_{E \cup A}$
to simpler forms of rewriting just modulo $A$,
where both equality modulo $A$ and matching modulo $A$ are
decidable.

The first condition is that $E$ should be {\em terminating}
and {\em ground confluent} modulo $A$~\cite{dershowitz90}. This
means that in the rewrite theory $\thname{R}_{E/A} = (\Sigma, A, E)$,
(i) all rewrite sequences terminate, that is, there are no infinite
sequences of the form $[t_1]_A\longrightarrow_{\thname{R}_{E/A}} [t_2]_A
\cdots [t_n]_A \longrightarrow_{\thname{R}_{E/A}} [t_{n+1}]_A\cdots $, and
(ii) for each $[t]_A \in T_{\Sigma/A}$ there is a {\em unique} 
$A$-equivalence class $[{\sf can}_{E/A}(t)]_A \in T_{\Sigma/A}$ called
the {\em $E$-canonical form} of $[t]_A$ modulo $A$ such that there
exists a terminating sequence of zero, one, or more steps
$[t]_A \longrightarrow^*_{\thname{R}_{E/A}} [{\sf can}_{E/A}(t)]_A$.

The second condition is that the rules $R$ should be {\em coherent}~\cite{viry-tcs}
relative to the equations $E$ modulo $A$. This precisely means that, if
we decompose the rewrite theory $\thname{R} = (\Sigma, E \cup A , R)$ into
the simpler theories $\thname{R}_{E/A} = (\Sigma, A, E)$ and 
$\thname{R}_{R/A} = (\Sigma, A , R)$, which have decidable rewrite
relations $\longrightarrow_{\thname{R}_{E/A}}$ and
$\longrightarrow_{\thname{R}_{R/A}}$ because of the assumptions
of $A$, then for each $A$-equivalence class $[t]_A$ such that
$[t]_A \longrightarrow_{\thname{R}_{R/A}}[t']_A$ we can always find
a corresponding rewrite $[{\sf can}_{E/A}(t)]_A \longrightarrow_{\thname{R}_{R/A}}
[t'']_A$ such that $[{\sf can}_{E/A}(t')]_A = [{\sf can}_{E/A}(t'')]_A$.
Intuitively, coherence means that we can always adopt the strategy of
first simplifying a term to canonical form with $E$ modulo $A$, and then
apply a rule with $R$ modulo $A$ to achieve the effect of rewriting
with $R$ modulo $E \cup A$.

The conceptual distinction between equations and rules has
important consequences when giving the 
rewriting logic semantics of a language $L$ as a rewrite theory
$\thname{R}_L=(\Sigma_L,E_L \cup A_L,R_L)$.
Rewriting logic's {\em abstraction dial}~\cite{meseguerrosu07} captures precisely this 
conceptual distinction. One of the key features of Structural Operational Semantics
is that it provides a step-by-step formal description
of a language's evaluation mechanisms~\cite{Plotkin:NatSemTR}. Setting 
the level of abstraction in which the interleaving
behavior of the evaluations in $L$ is observable, corresponds 
to the special case in which the dial is {\em turned down to its minimum position}
by having $E_L \cup A_L = \emptyset$. The abstraction dial can also be {\em turned
up to its maximal position} as the special case in which $R_L = \emptyset$,
thus obtaining an equational semantics of the language. In general,
we can make a specification as {\em abstract as
we want} by identifying a subset $R'_L\subseteq R_L$ such that
the rewrite theory $(\Sigma_L, (E_L \cup R'_L)\cup A_L , R_L \setminus R'_L)$
satisfies the executability conditions aforementioned. We refer
the reader to~\cite{verdejo06,meseguerrosu07,serbanuta09} for an
in-depth presentation of the relationship between structural operational
semantics and rewriting logic semantics, and the use of equations
and rules to capture in rewriting logic the dynamic behavior of
language semantics.

The conceptual distinction between equations and rules also 
has important practical consequences
for {\em program analysis}, because it affords massive {\em state space
reduction} which can make formal analyses such as breadth-first search
and model checking enormously more efficient. Because of state-space
explosion, such analyses could easily become infeasible if we were
to use a specification in which all computation steps are described
with rules.

\section{A Rewriting Library for Synchronous Relations}
\label{sec.library}

When designing a programming language, it is useful to be able to
define its semantic relation, to formally prove properties of this
relation and to execute it on particular programs. However, defining
such a semantic relation and formally reasoning about it is generally
difficult, time consuming, and error-prone. This would be a major
endeavor if it had to be done from scratch for each language.
Moreover, since programming languages tend to evolve constantly, tools
must allow reusing parts of former developments to support rapid yet
correct prototyping.

Fortunately, small-step operational semantic relations 
are, in general, built from simple
relations with a limited number of operations, such as
reflexive-transitive extension, reduction to normal form, parallel
extension, etc.  As a minimum, the framework should include a library
containing the definitions of these operations and formal proofs of
their properties. This will considerably reduce the amount of work
needed to define the semantic relation of {\em particular} programming
languages and to formally prove their properties. Defining the semantic relation
of synchronous languages requires defining the synchronous 
extension of an atomic execution relation, an operation that 
has been much less studied formally than other relation operations 
such as the reflexive-transitive extension or the parallel
extension.

We present in this section a first attempt to design a framework for
rapid yet correct prototyping of semantic relations, in particular of
synchronous languages.
This framework allows one to define semantic relations, to execute
them on particular programs and to formally prove some of their
properties using general theorems about the operations that
permit to build relations from relations.
We have been
experimenting with this framework using various versions of the PLEXIL language 
(see Section \ref{sec.plexilrls}). 

The definitions and properties presented in Section~\ref{sec.setr} have been
developed in PVS. The Maude engine is used for executing the semantic
relations on particular programs. The full development
of the framework,
including the formal semantics of PLEXIL, is available
from~\url{http://research.nianet.org/fm-at-nia/PLEXIL}.

\subsection{Set Relations and Determinism}
\label{sec.setr}

Let $\lra$ be a binary relation on a set $T$. 
We say that $a \in T$ is a {\em redex} if there exists $a' \in T$ such
that $a \lra a'$, and that it is a {\em normal form} otherwise.
We denote by $\lrn{0}$, $\lrn{n}$, and $\lrn{*}$, the {\em identity} 
relation, {\em $n$-fold composition}, and {\em reflexive-transitive closure} 
of $\lra$, respectively.

In addition to the above relations, we also define the 
{\em normalized reduction relation} $\lrn{\downarrow}$ of 
$\lra$.

\begin{definition}[Normalized reduction]
$a \lrn{\downarrow} a'$ if and only if $a \lrn{*} a'$ and $a'$ is a 
normal form.
\end{definition}

Henceforth, we assume that the relation $\lra$ is defined on sets 
over an abstract type $T$, i.e., 
$\lra\ \subseteq \mathcal{P}(T)\times\mathcal{P}(T)$. 
We define the {\em asynchronous} extension of $\lra$, denoted $\async$,
as the congruence closure of $\lra$ and the {\em parallel} extension of
$\lra$, denoted $\para$, as the parallel closure of $\lra$. 

\begin{definition}[Asynchronous extension]
$a \async a'$ if and only if there exist  
sets $b$ and $b'$ such that $b \subseteq a$, $b \neq \varnothing$, 
$b \lra b'$ 
and $a' = (a \setminus b) \cup b'$.
\end{definition}

\begin{definition}[Parallel extension]
$a \para a'$ if and only if there exist $b_1,\ldots, b_n$, nonempty,
pairwise disjoint subsets of $a$, and sets $b'_1,\ldots, b'_n$ such that 
$b_i \lra b'_i$ and $a' = (a \setminus \bigcup_{i} b_i) \cup \bigcup_{i} b'_i$.
\end{definition}



The definition of a synchronous reduction requires the definition
of a strategy that selects the redexes to be synchronously reduced. 

\begin{definition}[Strategy]
A {\em strategy} is a function mapping elements
$a \in \mathcal{P}(T)$ into  $b_1, \ldots, b_n$, nonempty, pairwise 
disjoint subsets of $a$ such that all $b_i$ are redexes for $\lra$.
\end{definition}

\begin{definition}[Synchronous extension]
Let $s$ be a strategy, $a \syncr{s} a'$ if and only if there exist 
$b'_1, \ldots, b'_n$
such that $s(a) = \{b_1, \ldots, b_n\}$, $b_i \lra b'_i$ and $a' = (a
\setminus \bigcup_{i} b_i) \cup \bigcup_{i} b'_i$.
\end{definition}


A natural way of defining strategies is via priorities. A {\em priority}
is a function $p$ that maps elements $a \in \mathcal{P}(T)$ into
natural numbers. 

\begin{definition}[Maximal redex]
Let $a \in \mathcal{P}(T)$ and let $p$ be a priority function. 
A nonempty subset $b$ of $a$ is said to be a 
{\em maximal redex} of $a$ if it is a redex, and for all nonempty 
subsets $c$ of $a$ such that $c$ is a redex, $c \neq b$ and 
$c \cap b \neq \varnothing$, we have $p(b) > p(c)$. 
By construction, the set of maximal redexes of a set are pairwise disjoint.
The {\em maximal redexes strategy} is the function that,
given a priority function, 
maps elements $a \in \mathcal{P}(T)$ into the set of its maximal redexes.
\end{definition}


In addition to the definition of the relation operators presented
here, our library includes formal proofs of
properties related to determinism and compositionality for abstract set relations. 
In this paper, we will focus on determinism as this property is fundamental to 
the specification of synchronous relations in rewriting logic.

\begin{definition}[Determinism]
A binary relation $\lra$ defined on a set $T$ is said to be {\em
deterministic} if for all $a$, $a'$ and $a''$ in $T$, $a \lra a'$ and $a
\lra a''$ implies $a' = a''$.
\end{definition}

Determinism is a stronger property than confluence, i.e.,
a deterministic relation is also confluent, but a confluent relation is not necessarily deterministic.

\begin{proposition}[Determinism of $\lrn{n}$, $\lrn{\downarrow}$ and 
$\syncr{s}$]
\label{prop.detsync}
If the relation $\lra$ is deterministic, then so are the 
relations $\lrn{n}$, $\lrn{\downarrow}$, and
$\syncr{s}$.
\end{proposition}

In contrast, even if the relation $\lra$ is deterministic, 
the relations $\lrn{*}$, $\async$ and $\lrn{\parallel}$ are not always
deterministic. 

\subsection{Executing Semantic Relations}
\label{sec.exec}

Executing the semantic relation of a programming language is desirable
during the design phase of the language. In
particular, it allows the designer of the features to experiment with
different semantic variants of the language before implementing
them. 
 	
Rewrite systems are a computational way of defining binary
relations. Since our formalism is based on set relations, we consider
rewrite systems on an algebra of {\em terms} of type $T$ modulo
associativity, commutativity, identity, and idempotence: the
basic axioms for the union of sets.  We denote
the equality on terms of this algebra by $\eqaci$.  The relation
$\lra$ defined by a rewrite system $\cal R$ is defined as follows.

\begin{definition}[Relation defined by a rewrite system]
$a \lra b$ if and only if there exists a rewrite rule $l \rew r$ in ${\cal R}$ 
and a substitution $\sigma$ such that $a \eqaci \sigma l$ and $b \eqaci \sigma r$. 
\end{definition}
 
We remark that the previous definition uses the substitution closure 
of the rewrite system, rather than the more traditional 
definition based on the congruence closure. For example, if we consider the
rewrite system
\begin{eqnarray*}
A(x) & \rew &B(x),
\end{eqnarray*}
we have that $A(0) \lra B(0)$ and $A(1) \lra B(1)$. On the other hand,
$A(0),A(1)$ is not a redex for $\lra$.

The synchronous extension of a relation $\lra$ challenges the standard 
asynchronous interpretation of rewrite systems. Consider again the 
previous example. The asynchronous extension of $\lra$ 
defined in Section~\ref{sec.setr}, which indeed encodes the congruence closure,
relates $A(0),A(1) \async B(0),A(1)$ and  $A(0),A(1) \async A(0),B(1)$. However,
it does not relate $A(0),A(1)$ to $B(0),B(1)$, which corresponds to the 
parallel reduction of both $A(0)$ and $A(1)$.
In this particular case, we have that 
$A(0),A(1) \async B(0),A(1) \async B(0),B(1)$. 

We remark that if $a \syncr{s} b$,
for a strategy $s$, then $a {\async}^* b$. However, in order to select the
redexes to be reduced, we need additional machinery. In particular, we need
to keep a log book of redexes that need to be reduced
and redexes that have been already reduced.
We propose the following procedure to implement 
in an asynchronous rewrite engine, such as Maude, the synchronous 
extension of a relation for a strategy.

\begin{definition}[Serialization procedure]
Let $\rightarrow$ be a relation and $s$ a strategy.
Given a term $a \in \mathcal{P}(T)$, we compute a term $b$ as follows.
\begin{enumerate}
\item Reduce the pair $\langle \bigcup s(a)\ ;\ \varnothing\rangle$ to a normal
form $\langle \varnothing\ ;\ a'\rangle$ using the following rewrite system:
\begin{eqnarray*}
\langle a_i,c\ ;\ d \rangle&\rew&\langle c\ ;\ a'_i,d\rangle, 
\end{eqnarray*}
where $a_i \lra a'_i$. 
\item  The term $b$ is defined as $(a \setminus \bigcup s(a)) \cup a'$.
\end{enumerate}
\end{definition}
Since a strategy is a set of redexes, and this set is finite,
the procedure is well-defined, i.e., it always terminates and returns
a term. However, the procedure is not necessarily deterministic.

In our previous example, we want to apply 
the procedure to $A(0),A(1),B(1)$ using the maximal redexes 
strategy $\texttt{max}_p$ (assuming that all terms have the 
same priority). Since 
$\texttt{max}_p(\{A(0),A(1),B(1)\}) = \{A(0),A(1)\}$, we have to
reduce the pair $\langle A(0),A(1)\ ;\ \varnothing\rangle$ to its
normal form $\langle \varnothing\ ;\ B(0),B(1)\rangle$. Then, 
we compute $\{A(0),A(1),B(1)\}\ \setminus\ \{A(0),A(1)\}\ \cup\ \{B(0),B(1)\}$,
which is equal to $B(0),B(1)$. We check that 
$A(0),A(1),B(1) \syncr{s} B(0),B(1)$. 

\begin{theorem}[Correctness of serialization procedure]
The serialization procedure is {\em sound}, i.e., if the procedure
returns $b$ from $a$, then $a \syncr{s} b$. Furthermore, if $\lra$ is
deterministic, the procedure is {\em complete}, i.e., if $a \syncr{s}
b$ then the procedure returns $b$ from $a$.
\end{theorem}
\begin{proof}
\begin{description}
\item[Soundness]  Assume that the procedure returns 
        $b = a \setminus \bigcup s(a) \cup a'$
	from~$a$. We have to prove that $a \syncr{s} b$.
        Let $s(a) = \{a_1, \dots, a_n\}$, 
        where $a_i \subseteq a$,
	for $1 \leq i \leq n$.
        From the definition of a strategy, 
        the elements in $s(a)$ are pairwise disjoint. Then, from the procedure,
	$a'= a_1',\dots,a_n'$, where $a_i \rightarrow a_i'$, for
	$1 \leq i \leq n$.
	Let $c \subseteq a$ be such that none of the subsets of $c$ is in
        $s(a)$. Then, $a$ has the form $a = a_1, \dots , a_n ,c$. Hence,
        $b= a_1',\dots,a_n',c$. By definition of $\syncr{s}$, we have that
        $a_1, \dots , a_n , c \syncr{s} a_1', \dots, a'_n , c$.
\item[Completeness] In this case, it suffices to note that by 
     Proposition~\ref{prop.detsync}, if $\lra$ is deterministic, then
     $\syncr{s}$ is deterministic. Therefore, the normal form of 
     $\langle \bigcup s(a)\ ;\ \varnothing\rangle$ is unique and the
     procedure returns a unique $b= a \setminus \bigcup s(a) \cup a'$. This
     $b$ is the only term that is related to $a$ in the relation 
     $\syncr{s}$.
\end{description}
\end{proof}

\section{Rewriting Logic Semantics of PLEXIL}
\label{sec.plexilrls}

The framework presented in Section~\ref{sec.library} is abstract with respect to the elements in the set $T$ and the basic
set relation $\lra$. If we consider that $T$ is a set of PLEXIL nodes and $\lra$ is PLEXIL's atomic relation, we can
deduce by Proposition~\ref{prop.detsync} that, since PLEXIL's atomic relation is deterministic~\cite{DMP07ICAPSW} , 
PLEXIL's micro and quiescence relations are deterministic as well. Therefore, we can use the serialization procedure
presented  in Section~\ref{sec.exec} to implement a sound and complete formal interpreter of PLEXIL in Maude.  

In this section, we describe in detail the specification of such an interpreter.
We only discuss the atomic and micro relations since they are
the most interesting ones for validating the synchronous semantics
of PLEXIL. 
More precisely, we present the rewrite theory 
$\thname{R}_\plexil = (\Sigma_\plexil,E_\plexil \cup A_\plexil,R_\plexil)$,
specifying the rewriting logic semantics for PLEXIL's atomic and micro
relations.  We use the determinism property of PLEXIL's atomic relation 
to encode it as the computation rules in $E_\plexil$ because it
yields a confluent equational specification. Consequently,
the serialization procedure for PLEXIL's synchronous semantics into 
rewriting logic can be defined equationally,
thus avoiding the interleaving semantics associated with rewrite rules
in Maude. Of course, due to the determinism property
of the language, one can as well turn up the ``abstraction dial''
to its maximum
by 
making the rewrite rules $R_\plexil$ into
computational rules. This will result in a faster interpreter,
for example. Nevertheless, 
we are interested in PLEXIL
semantics at the observable level of the micro relation.
Therefore, in the rewrite theory $\thname{R}_\plexil$: 
(i) the equational theory $(\Sigma_\plexil,E_\plexil \cup
A_\plexil)$ defines the semantics of the atomic relation
and specifies the serialization procedure for the synchronous
semantics of PLEXIL, and (ii) the rewrite
rules $R_\plexil$ define the semantics of the micro relation.

In this section we assume the reader is familiar with the
syntax of Maude~\cite{maude-book}, which is very close to standard
mathematical notation.

\subsection{PLEXIL Syntax}
\label{sec.plexil.syntax}

A PLEXIL \emph{plan} is a tree of \emph{nodes}
representing a hierarchical decomposition of tasks. 
The interior nodes in a plan provide the control structure and 
the leaf nodes represent primitive actions. The purpose of each
node determines its {\em type}: {\tt List} nodes group other nodes and provide
scope for local variables,
{\tt Assignment} nodes assign values to variables (they also have
a {\em priority}, which serves to solve race conditions between assignment
nodes), {\tt Command} nodes
represent calls to commands, and {\tt Empty} nodes do nothing.
Each PLEXIL node has 
\emph{gate conditions} and \emph{check
conditions}. The former specify when the node should start executing,
when it should finish executing, when it should be repeated, and when
it should be skipped. Check conditions specify flags to detect when 
node execution fails due to violations of pre-conditions, post-conditions, or 
invariants. Declared \emph{variables} in nodes have lexical scope, that
is, they are accessible to the node and all its descendants, but not
siblings or ancestors. The \emph{execution status} of a node is given by
status such as \texttt{Inactive}, \texttt{Waiting}, \texttt{Executing},
etc. The {\em execution state} of a plan consists of (i) the \emph{external state} 
corresponding to a set of {\em environment variables} accessed 
through \emph{lookups} on environment variables, and (ii) the {\em internal state}
which is a set of nodes and (declared) variables.

\begin{figure}[htp]
\hspace{2.3cm}\begin{tabular}{|c|}
\hline \\
\begin{minipage}{11cm}
%
\begin{alltt}{\scriptsize
List SafeDrive \{
  int pictures = 0;
  End:
    LookupOnChange(WheelStuck) == true  OR  pictures == 10;
  List Loop \{
    Repeat-while:
      LookupOnChange(WheelStuck) == false;
    Command OneMeter \{
      Command: Drive(1);
    \}
    Command TakePic \{
      Start: OneMeter.status == FINISHED  AND  pictures < 10;
      Command: TakePicture();
    \}
    Assignment Counter \{
      Start: TakePic.status == FINISHED;
      Pre: pictures < 10;
      Assignment: pictures := pictures + 1;
    \}
  \}
\}}

\end{alltt}
\end{minipage} \\
\hline
\end{tabular}
	\caption{SafeDrive: A PLEXIL Plan Example}
	\label{fig:safedrive}
\end{figure}

Figure~\ref{fig:safedrive} illustrates with a simple example
the standard syntax of PLEXIL. 
In this particular example, the plan tasks are represented by the root node 
{\tt SafeDrive}, the interior node  {\tt Loop}, and the leaf nodes {\tt OneMeter}, 
{\tt TakePic} and {\tt Counter}. 
{\tt OneMeter} and {\tt TakePic} are, for example, nodes of type {\tt Command}.
The node {\tt Counter} has two different conditions: {\tt Start} is a gate condition
constraining the execution of the assignment to start only when
the node {\tt TakePic} is in state {\tt Finished}, while {\tt Pre} is a 
check condition for the number of pictures to be less than 10. The
internal state of the plan at a particular moment is represented 
by the set of all nodes of the plan, plus the value of the variable 
{\tt pictures}, while the external state of the plan contains the
(external) variable {\tt WheelStuck}.

The external state of a plan is defined in the functional module 
{\tt EXTERNAL-STATE-SYNTAX}. The sort {\bf ExternalState} represents  
sets of elements of sort {\bf Pair}, each of the form $(name,value)$;  
we assume that the sorts {\bf Name} and {\bf Value}, specifying names  
and values, respectively, have been defined previously in the functional  
modules {\tt NAME} and {\tt VALUE}, respectively.

\begin{alltt}{\scriptsize
fmod EXTERNAL-STATE-SYNTAX is 
  protecting Name .  protecting Value .
  sort Pair . 
  op (_,_) : Name Value -> Pair . 
  sort ExternalState . 
  subsort Pair < ExternalState . 
  op mtstate : -> ExternalState . 
  op _,_ : ExternalState ExternalState -> ExternalState [assoc comm id: mtstate] . 
  eq ES:ExternalState , ES:ExternalState = ES:ExternalState .  
endfm
}\end{alltt}
 
The internal state of a plan is represented with the help of Maude's built-in
{\tt CONF} module supporting object based programming. The internal state
has the structure of a {\em set} made up of objects and messages, called
{\em configurations} in Maude, where the objects represent the nodes and 
(declared) variables of a plan. Therefore, we can view the infrastructure of the internal state as 
a configuration built up by a binary set union operator with empty syntax, i.e., juxtaposition, as $\_ \_ : 
{\bf Configuration \times Configuration \longrightarrow}$ ${\bf Configuration}$. The operator $\_\_$ is declared
to satisfy the structural laws of associativity and commutativity and to have
identity {\tt mtconf}. Objects and messages are singleton set configurations and
belong to subsorts ${\bf Object, Msg < Configuration}$, so that more complex configurations
are generated out of them by set union. An {\em object}, representing a node
or a (declared) variable, in a given configuration is represented as a term 
$\langle O : C \; | \; a_1 : v_1 , \dots , a_n : v_n \rangle$, where $O$ is the object's
name or identifier (of sort {\bf Oid}), $C$ is its class (of sort {\bf Cid}), 
the $a_i$'s are the names of the object's
{\em attribute identifiers}, and the $v_i$'s are the corresponding {\em values}.
The set of all the attribute-value pairs of an object state (of sort {\bf Attribute})
is formed by repeated application of the binary union operator $\_,\_$ which also obeys structural laws
of associativity, commutativity, and identity, i.e., the order of the attribute-value
pairs of an object is immaterial. The internal state of a plan is defined in the 
functional module {\tt INTERNAL-STATE-SYNTAX} by extending the sort {\bf Configuration};
the sorts {\bf Exp} and {\bf Qualified}, which we assume to be defined, are used 
to specify expressions and qualified names, respectively.

% ...
\begin{alltt}{\scriptsize
fmod INTERNAL-STATE-SYNTAX is
  extending CONFIGURATION . protecting EXP .
  protecting QUALIFIED .
  subsort Qualified < Oid .                      --- Qualified elements are object identifiers
  ops List Command Assignment Empty : -> Cid .   --- Types of nodes
  sort Status .
  ops Inactive Waiting Executing Finishing Failing Finished IterationEnded Variable : -> ExecState .
  sort Outcome .
  ops None Success Failure : -> Outcome .
  op  status: : Status -> Attribute .      --- Status of execution
  op  outcome: : Outcome -> Attribute .    --- Outcome of execution
  ops start: skip: repeat: end: : Exp -> Attribute .  --- Gate conditions
  ops pre: post: inv: : Exp -> Attribute . --- Check conditions
  op  command: : Exp -> Attribute .        --- Command of a command node
  op  assignment: : Exp -> Attribute .     --- Assignment of an assignment node
  ops initval actval: Exp -> Attribute .   --- Initial and actual values of a variable node
  ...
endfm
}\end{alltt}

Using the infrastructure in {\tt INTERNAL-STATE-SYNTAX}, the internal
state of SafeDrive in Figure~\ref{fig:safedrive}, is represented 
by the configuration in Figure~\ref{fig:safedriveinmaude}. Observe
that the sort {\bf Qualified} provides qualified names by means of
the operator ${\bf \_.\_ : Qualified \times Qualified \longrightarrow Qualified}$,
which we use to maintain the hierarchical structure of the plans. The
dots at the end of each object represent the object's attributes
that are not explicitly defined by the plan but that are always
present in each node such as the status or the outcome. There is a
``compilation procedure'' from PLEXIL plans to their corresponding
representation in Maude, that we do not discuss in this paper, which
includes all implicit elements of a node as attributes of the object
representation of the node.

\begin{figure}[htp]
  \centering
\begin{tabular}{|c|}
\hline \\
\begin{minipage}{14cm}
\begin{alltt}{\scriptsize
< SafeDrive : List | end: LookupOnChange(WheelStuck) == true OR pictures == 10, ... >
< Loop . SafeDrive : List | repeat: LookupOnChange(WheelStuck) == false , ... >
< OneMetter . Loop . SafeDrive : Command | command: Drive(1), ... >
< TakePic . Loop . SafeDrive : Command | start: OneMeter.Status == Finished and pictures < 10, 
  command: TakePicture(), ... >
< Counter . Loop . SafeDrive : Assignment | pre: pictures < 10, 
  assignment: pictures := pictures + 1, ... >
< pictures . SafeDrive : Memory | initval: 0, actval: 0 >

}\end{alltt}
\end{minipage} \\
\hline
\end{tabular}
	\caption{SafeDrive in $\thname{R}_\plexil$}
	\label{fig:safedriveinmaude}
\end{figure}

We are now ready to define the sort {\bf State} representing 
the execution state of the plans in the functional module 
{\tt STATE-SYNTAX}, by importing the syntax of external and
internal states:

\begin{alltt}{\scriptsize
fmod STATE-SYNTAX is
  pr EXTERNAL-STATE-SYNTAX .
  pr INTERNAL-STATE-SYNTAX .
  sort State .
  op _|-_ : ExternalState Configuration -> State .
endfm
}\end{alltt}

We adopt the syntax $\Gamma\vdash \pi$ to represent
the execution state of the plans, where $\Gamma$ and $\pi$ are the 
external and internal states, respectively.

\subsection{PLEXIL Semantics}
\label{sec.plexil.semantics}

PLEXIL execution is driven by external events. The set of events includes
events related to lookup in conditions, e.g., changes in the value of an
external state that affects a gate condition, acknowledgments that a command
has been initialized, reception of a value returned by a command, etc.
We focus on the execution semantics of PLEXIL specified in terms of node 
states and transitions between node states that are triggered by condition 
changes (atomic relation) and its synchronous closure under the maximal redexes 
strategy (micro relation).
PLEXIL's atomic relation consists of 42 rules, indexed  
by the type and the execution status of nodes into a dozen groups. Each group  
associates a priority to its set of rules which defines a linear order on the 
set of rules. 

The {\em atomic relation} is defined by $(\Gamma,\pi)\vdash P \longrightarrow_{a} P'$, 
where $P \subseteq \pi$. For
instance, the four atomic rules corresponding to the transitions from {\tt Executing}
for nodes of type {\tt Assignment} are depicted in Figure~\ref{fig:atomic}.
Rule $r_3$ updates the status and the outcome of node {\tt A}
to the values {\tt IterationEnded} and 
{\tt Success}, respectively, and the variable $x$ to the value {\tt v}, i.e.,
the value of the expression {\tt e} in the state $\pi$, whenever the 
expressions associated with the gate condition {\tt End}
and the check condition {\tt Post} of node {\tt A} both evaluate to {\tt true}
in $\pi$. 
In rule $r_1$, {\tt AncInv(A)} is a predicate, parametric
in the name of nodes, stating that none of the ancestors of {\tt A} has changed
the value associated with its invariant condition to {\tt false}. The value $\bot$
represents the special value ``Unknown''.
We use $(\Gamma,\pi)\vdash e \leadsto v$ to denote
that expression $e$ evaluates to value $v$ in state $\Gamma \vdash \pi$;
by abuse of notation, we write $(\Gamma,\pi)\vdash e \not \leadsto v$ to
denote that expression $e$ does not evaluate to value $v$ in $(\Gamma,\pi)$.

\begin{figure}[htp]
  \centering
  \fbox{
{\small
\begin{tabular}{c}
 $ \inferrule*[right = $r_1$]
   {(\Gamma,\pi) \vdash {\tt AncInv(A) \leadsto false} \\ {\tt A.body\; =\; x\,:=\,e}\\\\ 
    {\tt A.type = Assignment} \\ {\tt A.status = Executing}
   }
   {(\Gamma,\pi)\vdash {\tt Node\;A} \longrightarrow_a \;{\tt Node \; A \; with\;[status = Finished\,,\,outcome=Failure\,,\,x=\bot]} } 
 $ \\ \\

 $ \inferrule*[right = $r_2$]
   {(\Gamma,\pi) \vdash {\tt A.Invariant \leadsto false} \\ {\tt A.body\; =\; x\,:=\,e}\\\\ 
    {\tt A.type = Assignment} \\ {\tt A.status = Executing}
   }
   {(\Gamma,\pi)\vdash {\tt Node\;A} \longrightarrow_a \;{\tt Node \; A \; with\;[status = IterationEnded\,,\,outcome=Failure\,,\,x=\bot]} } 
 $ \\ \\

 $ \inferrule*[right = $r_3$]
   {(\Gamma,\pi) \vdash {\tt A.End \leadsto true} \\  {\tt A.body\; =\; x\,:=\,e}\\\\ 
    (\Gamma,\pi) \vdash {\tt A.Post \leadsto true} \\ (\Gamma,\pi) \vdash {\tt e \leadsto v} \\\\
    {\tt A.type = Assignment} \\ {\tt A.status = Executing}
   }
   {(\Gamma,\pi)\vdash {\tt Node\;A} \longrightarrow_a \;{\tt Node \; A \; with\;[status = IterationEnded\,,\,outcome=Success\,,\,x=v]} }
 $ \\ \\

 $ \inferrule*[right = $r_4$]
   {(\Gamma,\pi) \vdash {\tt A.End \leadsto true} \\ (\Gamma,\pi) \vdash {\tt A.Post \not\leadsto true} \\\\
    {\tt A.type = Assignment} \\ {\tt A.status = Executing}
   }
   {(\Gamma,\pi)\vdash {\tt Node\;A} \longrightarrow_a \;{\tt Node \; A \; with\;[status = IterationEnded\,,\,outcome=Failure]} } 
 $ \\ \\
 $\{r_4 < r_3 < r_2 < r_1\}$ \\
\end{tabular}
}
	}
	\caption{Atomic rules corresponding to the transitions from {\tt Executing}
for nodes of type {\tt Assignment}}
	\label{fig:atomic}
\end{figure}

The relation $r < s$ between the labels of two different rules specifies
that the rule $r$ is only applied when the second rule $s$ cannot be
applied. That is, the binary relation on rules defines the {\em order 
of their application} when deriving atomic transitions. So, a rule $r$
can be used to derive an atomic transition if all its premises are valid
and no rule higher than $r$ (in its group) is applicable. 
In the case
of PLEXIL's atomic relation, the binary relation $<$ on rules is a linear 
ordering. This linearity is key to the determinism of PLEXIL
(see~\cite{DMP08}). 

The {\em micro relation} $\Gamma \vdash \pi \longrightarrow_m \pi'$, 
the synchronous closure of the atomic relation under the maximal redexes strategy, 
is defined as:
\[
\inferrule*[right = Micro]
 {(\Gamma,\pi) \vdash P_1 \longrightarrow_a P_1' \\\\
  \dots \\\\
  (\Gamma,\pi) \vdash P_n \longrightarrow_a P_n'
 }
 {\Gamma \vdash \pi \longrightarrow_m (\pi \setminus \bigcup_{1\leq i \leq n} P_i) \; \cup \; \bigcup_{1\leq i \leq n} P_i' }
\]
where $M_\pi = \{ P_1,\dots,P_n\}$ is the set of nodes and variables 
in $\pi$ that are affected by the micro relation. If two different
processes in $\pi$, say {\tt A} and {\tt B}, write to the same variable,
only the update of the process with higher priority is considered 
(assignment nodes have an associated priority always), e.g., {\tt A}
if {\tt A.priority > B.priority}, {\tt B} if {\tt B.priority >
A.priority}, and none otherwise.

In order to specify the PLEXIL semantics in Maude, we first define 
the infrastructure for the serialization procedure in 
the functional module {\tt SERIALIZATION-INFRASTRUCTURE}.

% ...
\begin{alltt}{\scriptsize
fmod SERIALIZATION-INFRASTRUCTURE is
  inc STATE-SYNTAX .
  ...
  op  [_:_|_] : Oid Cid AttributeSet -> Object .     --- New syntactic sugar for objects
  op  updateStatus    : Qualified Status  -> Msg .   --- Update status message
  op  updateOutcome   : Qualified Outcome -> Msg .   --- Update outcome message
  op  updateVariable  : Qualified Value   -> Msg .   --- Update variable message
  ...
  ops  applyUpdates unprime : State -> State .       --- Application of updates and `unpriming'
  var \(\Gamma\) : ExternalState .   var \(\pi\) : InternalState .  var A : Oid .  var C : Cid .
  var Att : AttributeSet .  vars S S' : Status .     var St : State .
  eq  applyUpdates( \(\Gamma \vdash\) [ A : C | status: S , Att ] updateStatus(A , S') \(\pi\) )
   =  applyUpdates( \(\Gamma \vdash\) [ A : C | status: S' , Att ] \(\pi\) ) . 
  ...
  eq  applyUpdates(St) = St [owise] .
  eq  unprime( \(\Gamma \vdash\) [ A : C | Att ] \(\pi\) ) =  unprime( \(\Gamma \vdash\) < A : C | Att > \(\pi\) ) .
  eq  unprime(St) = St [owise] .
endfm
}\end{alltt}

Following the idea
of the serialization procedure, we distinguish between unprimed and
primed redexes by using syntactic sugar for denoting objects in the Maude
specification: unprimed redexes are identified with the already defined
syntax for objects in the form of $\langle O : C \; | \; ...\rangle$ and primed
redexes are identified with the new syntax for objects in the form 
of $[ O : C \; | \; ...]$. We use messages, i.e., elements in the sort
{\bf Msg}, to denote the update actions associated with the reduction
rules for the atomic relation; we accumulate these messages in the
internal state of the execution state of the plans, i.e., we also use the
internal state in the spirit of the log book of the serialization
procedure. For example, the configuration 
{\tt updateStatus(A,IterationEnded) updateOutcome(A,Success) updateVariable(x,v)} 
corresponds to the update actions in the conclusion of rule $r_3$ 
in Figure~\ref{fig:atomic}. The functions {\tt applyUpdates} and
{\tt unprimes} apply all the collected updates in the internal state, and ``unprimes'' the
``primed'' nodes, respectively. In the specification above, it is shown how the status 
of a node is updated and how primed nodes become unprimed.

We give the equational serialization procedure in the general
setting in which we consider a linear ordering on the rules.

\begin{definition}[Equational serialization procedure (with priorities)] Let

\[\{r_i : (\Gamma,\pi) \vdash {\tt Node\; A} \; \longrightarrow_a \; {\tt Node \; A \; with \;[updates_i]\;}if\;{\tt C_i}\}_{1 \leq i \leq n}\] be the collection of atomic rules (in horizontal notation) defining the transition
relation for nodes of type $T$ in status $S$,
with $r_n < \cdots < r_i < \cdots < r_1$, where 
${\tt updates_i}$ is the set of update actions (the order
in the update actions is irrelevant) in the conclusion of $r_i$ and ${\tt C_i}$ is the set of
premises of $r_i$. The equational serialization procedure
is given by the following set of equations, in Maude notation, defining the 
function symbol, say, $r$:


\begin{alltt}{\scriptsize
  var \(\Gamma\) : ExternalState .   var A : Oid .   var S : Status .
  var \(\pi\) : Configuration .   var T : Cid .   var Attr : AttributeSet .
  op  r : State -> State .
  eq  r( \(\Gamma\) \(\vdash\) < A : T | status: S , Attr > \(\pi\) ) 
   =  if C\(\sb{1}\) == true then r( \(\Gamma\) \(\vdash\) [ A : T | status: S , Attr ] messages(updates\(\sb{1}\)) \(\pi\) )
      else if C\(\sb{2}\) == true 
           ... else if C\(\sb{n}\) == true then r( \(\Gamma\) \(\vdash\) [ A : T | status: S , Attr ] messages(updates\(\sb{n}\)) \(\pi\) )
                    else r( \(\Gamma\) \(\vdash\) [ A : T | status: S , Attr ] \(\pi\) ) fi
           fi ...
      fi .
  eq  r( \(\Gamma\) \(\vdash\) \(\pi\) ) = \(\Gamma\) \(\vdash\) \(\pi\) [owise] .
}\end{alltt}
 where ${\tt messages(updadates_i)}$ represents the message configuration
 associated with the update actions in the conclusion of rule $r_i$.
\end{definition}

The equational serialization procedure defines a fresh function symbol, say, 
$r : {\bf State \longrightarrow State}$. The first equation for $r$ tries to
apply the atomic rules in the given order, by first evaluating the condition
and then marking the affected node. If the condition evaluates to true, then
update messages are generated. The second equation, removes the function 
symbol $r$ when there aren't any more possible atomic reductions with the
rules $\{r_i\}$.


The atomic relation is defined
in the functional module {\tt ATOMIC-RELATION} by instantiating
the equational serialization procedure for each one of the 
twelve groups of atomic rules with a different function symbol
for each one.

Finally, the micro relation is defined by the rule {\tt micro} 
in the system module {\tt PLEXIL-RLS}, which materializes 
the rewrite theory $\thname{R}_\plexil$ in Maude:


\begin{alltt}{\scriptsize
mod PLEXIL-RLS is 
  pr ATOMIC-RELATION .
  pr SERIALIZATION-INFRASTRUCTURE .
  rl [micro] : \(\Gamma \vdash \pi\)  =>  \(\Gamma \vdash\) unprime(applyUpdates(a\(\sb{1}\)(...a\(\sb{12}\)( \(\Gamma \vdash \pi\) )...))) .
endm
}\end{alltt}

where ${\tt a_1,\dots,a_{12}}$ are the function symbols in {\tt ATOMIC-RELATION}
defining the serialization procedure for each one of the twelve groups of rules.




\section{Preliminary Results}
\label{sec.results}

We have used $\thname{R}_\plexil$ to 
validate the
semantics of PLEXIL against a wide variety of plan examples.
We report on the following two issues of the original PLEXIL semantics
that were discovered with the help of $\thname{R}_\plexil$:

\begin{enumerate}
	\item {\em Non-atomicity of the atomic relation.} A prior 
	version of the atomic rules $r_3$ and $r_4$ for 
	{\tt Assignment} nodes in state {\tt Executing},
	presented in Figure~\ref{fig:atomic}, introduced
	an undesired interleaving semantics for variable 
	assignments, which invalidated the synchronous
	nature of the language.
	\item {\em Spurious non-termination of plans.} Due to
	lack of detail in the original specification of some
	predicates, there were cases in which some transitions
	for {\tt List} nodes in state {\tt IterationEnded} would
	lead to spurious infinite loops.	
\end{enumerate}

Although the formal operational semantics of PLEXIL in~\cite{DMP08}
has been used to prove several properties of PLEXIL, 
neither one of the issues was previously found. As as matter of fact,
these issues do not compromise any of the proven properties
of the language.
Solutions to both 
issues were provided by the authors, and have been adopted
in the latest version of the formal PLEXIL semantics. We are currently 
using $\thname{R}_\plexil$ as the formal
interpreter of {\em PLEXIL's Formal Interactive
Visual Environment}~\cite{RMC09} (PLEXIL5), a prototype 
graphical environment that 
enables step-by-step execution of plans for scripted sequence
of external events, for further validation of  the language's intended semantics.

We have also developed a variant of  $\thname{R}_\plexil$ 
in which
the serialization procedure was implemented with rewrite rules,
instead of equations, and rewrite strategies. 
In general, $\thname{R}_\plexil$ outperforms that variant
by two orders of magnitude on average, and by three
orders of magnitude in some extreme cases.

The rewrite theory $\thname{R}_\plexil$ has approximately
1000 lines of code, of which 308 lines correspond to the module
{\tt ATOMIC-RELATION}. The rest corresponds to 
the syntax and infrastructure specifications.

\section{Related Work and Conclusion}
\label{sec.concl}

Rewriting logic has been used previously as a testbed for specifying
and animating the semantics of synchronous languages.
M. AlTurki and J. Meseguer~\cite{musab08} have studied the
rewriting logic semantics of the language Orc, which includes
a synchronous reduction relation. T. Serbanuta {\em et al.}~\cite{trian08}
define the execution of $P$-systems with structured
data with continuations. The focus of the former is to
use rewriting logic to study the (mainly) non-deterministic 
behavior of Orc programs, while the focus of the latter
is to study the relationship between $P$-systems
and the existing continuation framework for enriching
each with the strong features of the other. Our approach
is based more on exploiting the determinism of a synchronous
relation to tackle the problem associated with the interleaving
semantics of concurrency in rewriting logic.
P. Lucanu~\cite{lucanu09} studies the problem of the interleaving
semantics of concurrency in rewriting logic for 
synchronous systems from the perspective of $P$-systems.
The determinism property of the synchronous language
Esterel~\cite{Esterel00} was formally proven by O. Tardieu in~\cite{tardieu07}.

We have presented a rewriting logic semantics of
PLEXIL, a synchronous plan execution
language developed by NASA to support autonomous
spacecraft operations. The rewriting logic specification,
a formal interpreter and a semantic benchmark for validating
the semantics of the language,
relies on the determinism of PLEXIL's
atomic relation
and a serialization procedure that enables the specification of
a synchronous relation in an asynchronous computational model. 
Two issues in the original design of PLEXIL were found with the
help of the rewriting logic specification of the language: 
(i) there was an atomic rule with the potential 
to violate the atomicity of the atomic relation, thus voiding 
the synchronous nature of the language, and (ii) a set of rules introducing spurious
non-terminating executions of plans. We proposed solutions to these
issues that were integrated into the current semantics of
the language.

Although we have focused on PLEXIL, the formal framework that we have
developed is presented in a general setting of abstract set relations.
In particular, we think that this framework can be applied to other
deterministic synchronous languages. 
To the best of our knowledge there was no mechanized library
of abstract set relations suitable for the definition
and verification of synchronous relations; neither was there a
soundness and completeness proof of a serialization 
procedure for the simulation of synchronous relations
by rewrite systems.

To summarize, we view this work as (i) a step forward in bringing the use of
formal methods closer to practice, (ii) a contribution
to the modular and mechanized study of semantic relations, and 
(iii) yet another, but interesting contribution to the rewriting logic semantics
project.

We intend to continue our collaborative work with PLEXIL 
development team with the goal of arriving
at a formal environment for the validation of PLEXIL.
Such an environment would provide a rich formal tool
to PLEXIL enthusiasts for the experimentation, analysis
and verification of PLEXIL programs, which could then
be extended towards a rewriting-based PLEXIL implementation
with associated analysis tools. Part of our future work
is also to investigate the modularity of the equational serialization
procedure with prioritized rules.


\paragraph{\bf Acknowledgments.} This work was supported by the National
  Aeronautics and Space Administration at Langley Research Center under 
  the Research Cooperative Agreement No. NCC-1-02043 awarded to the National
  Institute of Aerospace, while the second author was resident at this institute.
  The third author was partially supported by NSF Grant IIS 07-20482.
  The authors would like to thank the members of the NASA's Automation for
  Operation (A4O) project and, especially, the PLEXIL development team led by Michael 
  Dalal at NASA Ames, for their technical support.

  \bibliographystyle{eptcs}
  \bibliography{biblio.bib}

\begin{thebibliography}{10}
\providecommand{\bibitemstart}[1]{\bibitem{#1}}
\providecommand{\bibitemend}{}
\providecommand{\bibliographystart}{}
\providecommand{\bibliographyend}{}
\providecommand{\url}[1]{\texttt{#1}}
\providecommand{\urlprefix}{Available at }
\providecommand{\bibinfo}[2]{#2}
\bibliographystart

\bibitemstart{musab08}
\bibinfo{author}{M.~AlTurki} \& \bibinfo{author}{J.~Meseguer}
  (\bibinfo{year}{2008}): \emph{\bibinfo{title}{Reduction Semantics and Formal
  Analysis of {Orc} Programs}}.
\newblock {\sl \bibinfo{journal}{Electr. Notes Theor. Comput. Sci.}}
  \bibinfo{volume}{200}(\bibinfo{number}{3}), pp. \bibinfo{pages}{25--41}.
\bibitemend

\bibitemstart{Esterel00}
\bibinfo{author}{G.~Berry} (\bibinfo{year}{2000}): \emph{\bibinfo{title}{The
  Foundations of {E}sterel}}.
\newblock In: {\sl \bibinfo{booktitle}{Proof, Language and Interaction: Essays
  in Honour of Robin Milner}}. \bibinfo{publisher}{MIT Press},
  \bibinfo{address}{Cambridge, MA, USA}, pp. \bibinfo{pages}{425--454}.
\bibitemend

\bibitemstart{bruni06}
\bibinfo{author}{R.~Bruni} \& \bibinfo{author}{J.~Meseguer}
  (\bibinfo{year}{2006}): \emph{\bibinfo{title}{Semantic foundations for
  generalized rewrite theories}}.
\newblock {\sl \bibinfo{journal}{Theor. Comput. Sci.}}
  \bibinfo{volume}{360}(\bibinfo{number}{1-3}), pp. \bibinfo{pages}{386--414}.
\newblock \urlprefix\url{http://dx.doi.org/10.1016/j.tcs.2006.04.012}.
\bibitemend

\bibitemstart{LustrePOPL87}
\bibinfo{author}{P.~Caspi}, \bibinfo{author}{D.~Pilaud},
  \bibinfo{author}{N.~Halbwachs} \& \bibinfo{author}{J.~A. Plaice}
  (\bibinfo{year}{1987}): \emph{\bibinfo{title}{LUSTRE: a declarative language
  for real-time programming}}.
\newblock In: {\sl \bibinfo{booktitle}{POPL '87: Proceedings of the 14th ACM
  SIGACT-SIGPLAN symposium on Principles of programming languages}}.
  \bibinfo{publisher}{ACM}, \bibinfo{address}{New York, NY, USA}, pp.
  \bibinfo{pages}{178--188}.
\bibitemend

\bibitemstart{maude-book}
\bibinfo{author}{M.~Clavel}, \bibinfo{author}{F.~Dur\'an},
  \bibinfo{author}{S.~Eker}, \bibinfo{author}{J.~Meseguer},
  \bibinfo{author}{P.~Lincoln}, \bibinfo{author}{N.~Mart\'{\i}-Oliet} \&
  \bibinfo{author}{C.~Talcott} (\bibinfo{year}{2007}):
  \emph{\bibinfo{title}{All About Maude - A High-Performance Logical
  Framework}}.
\newblock \bibinfo{publisher}{Springer LNCS Vol. 4350}, \bibinfo{edition}{1st}
  edition.
\bibitemend

\bibitemstart{dershowitz90}
\bibinfo{author}{N.~Dershowitz} \& \bibinfo{author}{J.~P. Jouannaud}
  (\bibinfo{year}{1990}): \emph{\bibinfo{title}{Rewrite Systems}}.
\newblock In: {\sl \bibinfo{booktitle}{Handbook of Theoretical Computer
  Science, Volume B: Formal Models and Sematics (B)}}. \bibinfo{publisher}{The
  MIT Press}, pp. \bibinfo{pages}{243--320}.
\bibitemend

\bibitemstart{DMP07ICAPSW}
\bibinfo{author}{G.~Dowek}, \bibinfo{author}{C.~Mu{\~{n}}oz} \&
  \bibinfo{author}{C.~P\u{a}s\u{a}reanu} (\bibinfo{year}{2007}):
  \emph{\bibinfo{title}{A Formal Analysis Framework for {PLEXIL}}}.
\newblock In: {\sl \bibinfo{booktitle}{Proceedings of 3rd Workshop on Planning
  and Plan Execution for Real-World Systems}}. pp. \bibinfo{pages}{45--51}.
\bibitemend

\bibitemstart{DMP08}
\bibinfo{author}{G.~Dowek}, \bibinfo{author}{C.~Mu{\~{n}}oz} \&
  \bibinfo{author}{C.~P\u{a}s\u{a}reanu} (\bibinfo{year}{2008}):
  \emph{\bibinfo{title}{A Small-Step Semantics OF {PLEXIL}}}.
\newblock \bibinfo{type}{Technical Report} \bibinfo{number}{2008-11},
  \bibinfo{institution}{National Institute of Aerospace},
  \bibinfo{address}{Hampton, VA}.
\bibitemend

\bibitemstart{plexil}
\bibinfo{author}{T.~Estlin}, \bibinfo{author}{A.~J\'{o}nsson},
  \bibinfo{author}{C.~P\u{a}s\u{a}reanu}, \bibinfo{author}{R.~Simmons},
  \bibinfo{author}{K.~Tso} \& \bibinfo{author}{V.~Verna}
  (\bibinfo{year}{2006}): \emph{\bibinfo{title}{{P}lan {E}xecution
  {I}nterchange {L}anguage ({PLEXIL})}}.
\newblock \bibinfo{type}{Technical Memorandum}
  \bibinfo{number}{TM-2006-213483}, \bibinfo{institution}{NASA}.
\bibitemend

\bibitemstart{lucanu09}
\bibinfo{author}{D.~Lucanu} (\bibinfo{year}{2009}):
  \emph{\bibinfo{title}{Strategy-Based Rewrite Semantics for Membrane Systems
  Preserves Maximal Concurrency of Evolution Rule Actions}}.
\newblock {\sl \bibinfo{journal}{Electr. Notes Theor. Comput. Sci.}}
  \bibinfo{volume}{237}, pp. \bibinfo{pages}{107--125}.
\bibitemend

\bibitemstart{unified-tcs}
\bibinfo{author}{J.~Meseguer} (\bibinfo{year}{1992}):
  \emph{\bibinfo{title}{Conditional Rewriting Logic as a Unified Model of
  Concurrency}}.
\newblock {\sl \bibinfo{journal}{Theoretical Computer Science}}
  \bibinfo{volume}{96}(\bibinfo{number}{1}), pp. \bibinfo{pages}{73--155}.
\bibitemend

\bibitemstart{meseguerrosu07}
\bibinfo{author}{J.~Meseguer} \& \bibinfo{author}{G.~Rosu}
  (\bibinfo{year}{2007}): \emph{\bibinfo{title}{The rewriting logic semantics
  project}}.
\newblock {\sl \bibinfo{journal}{Theor. Comput. Sci.}}
  \bibinfo{volume}{373}(\bibinfo{number}{3}), pp. \bibinfo{pages}{213--237}.
\newblock \urlprefix\url{http://dx.doi.org/10.1016/j.tcs.2006.12.018}.
\bibitemend

\bibitemstart{ORS92}
\bibinfo{author}{S.~Owre}, \bibinfo{author}{J.~Rushby} \&
  \bibinfo{author}{N.~Shankar} (\bibinfo{year}{1992}):
  \emph{\bibinfo{title}{{PVS}: A Prototype Verification System}}.
\newblock In: \bibinfo{editor}{Deepak Kapur}, editor: {\sl
  \bibinfo{booktitle}{11th International Conference on Automated Deduction
  (CADE)}}, {\sl \bibinfo{series}{Lecture Notes in Artificial Intelligence}}
  \bibinfo{volume}{607}. \bibinfo{publisher}{Springer-Verlag},
  \bibinfo{address}{Saratoga, NY}, pp. \bibinfo{pages}{748--752}.
\bibitemend

\bibitemstart{Plotkin:NatSemTR}
\bibinfo{author}{G.~D. Plotkin} (\bibinfo{year}{2004}): \emph{\bibinfo{title}{A
  structural approach to operational semantics}}.
\newblock {\sl \bibinfo{journal}{J. Log. Alg. Prog.}}
  \bibinfo{volume}{60-61}, pp. \bibinfo{pages}{17--139}.
\bibitemend

\bibitemstart{RMC09}
\bibinfo{author}{C.~Rocha}, \bibinfo{author}{C.~Mu{\~{n}}oz} \&
  \bibinfo{author}{H.~Cadavid} (\bibinfo{year}{2009}): \emph{\bibinfo{title}{A
  Graphical Environment for the Semantic Validation of a Plan Execution
  Language}}.
\newblock {\sl \bibinfo{journal}{IEEE International Conference on Space Mission
  Challenges for Information Technology}} \bibinfo{volume}{0}, pp.
  \bibinfo{pages}{201--207}.
\bibitemend

\bibitemstart{serbanuta09}
\bibinfo{author}{T.~Serbanuta}, \bibinfo{author}{G.~Rosu} \&
  \bibinfo{author}{J.~Meseguer} (\bibinfo{year}{2009}): \emph{\bibinfo{title}{A
  rewriting logic approach to operational semantics}}.
\newblock {\sl \bibinfo{journal}{Inf. Comput.}}
  \bibinfo{volume}{207}(\bibinfo{number}{2}), pp. \bibinfo{pages}{305--340}.
\bibitemend

\bibitemstart{trian08}
\bibinfo{author}{T.~Serbanuta}, \bibinfo{author}{G.~Stefanescu} \&
  \bibinfo{author}{G.~Rosu} (\bibinfo{year}{2008}):
  \emph{\bibinfo{title}{Defining and Executing P Systems with Structured Data
  in K}}.
\newblock In: \bibinfo{editor}{David~W. Corne}, \bibinfo{editor}{Pierluigi
  Frisco}, \bibinfo{editor}{Gheorghe Paun}, \bibinfo{editor}{Grzegorz
  Rozenberg} \& \bibinfo{editor}{Arto Salomaa}, editors: {\sl
  \bibinfo{booktitle}{Workshop on Membrane Computing}}, {\sl
  \bibinfo{series}{Lecture Notes in Computer Science}} \bibinfo{volume}{5391}.
  \bibinfo{publisher}{Springer}, pp. \bibinfo{pages}{374--393}.
\bibitemend

\bibitemstart{tardieu07}
\bibinfo{author}{O.~Tardieu} (\bibinfo{year}{2007}): \emph{\bibinfo{title}{A
  deterministic logical semantics for pure Esterel}}.
\newblock {\sl \bibinfo{journal}{ACM Trans. Program. Lang. Syst.}}
  \bibinfo{volume}{29}(\bibinfo{number}{2}), p.~\bibinfo{pages}{8}.
\bibitemend

\bibitemstart{verdejo06}
\bibinfo{author}{A.~Verdejo} \& \bibinfo{author}{N.~Mart\'{\i}-Oliet}
  (\bibinfo{year}{2006}): \emph{\bibinfo{title}{Executable structural
  operational semantics in Maude}}.
\newblock {\sl \bibinfo{journal}{J. Log. Algebr. Program.}}
  \bibinfo{volume}{67}(\bibinfo{number}{1-2}), pp. \bibinfo{pages}{226--293}.
\bibitemend

\bibitemstart{ue}
\bibinfo{author}{V.~Verna}, \bibinfo{author}{A.~J\'{o}nsson},
  \bibinfo{author}{C.~P\u{a}s\u{a}reanu} \& \bibinfo{author}{M.~Latauro}
  (\bibinfo{year}{2006}): \emph{\bibinfo{title}{Universal Executive and PLEXIL:
  Engine and Language for Robust Spacecraft Control and Operations}}.
\newblock In: {\sl \bibinfo{booktitle}{Proceedings of the American Institute of
  Aeronautics and Astronautics Space Conference}}.
\bibitemend

\bibitemstart{viry-tcs}
\bibinfo{author}{P.~Viry} (\bibinfo{year}{2002}):
  \emph{\bibinfo{title}{Equational rules for rewriting logic}}.
\newblock {\sl \bibinfo{journal}{Theoretical Computer Science}}
  \bibinfo{volume}{285}, pp. \bibinfo{pages}{487--517}.
\bibitemend

\bibliographyend
\end{thebibliography}
\end{document}